\documentclass[draftcls, 11pt, onecolumn]{IEEEtran}
\usepackage{amsmath,amsthm,amssymb,paralist,subfigure,cite,graphicx,color}
\usepackage{algorithm2e}
\usepackage[table]{xcolor}

\newtheorem{theorem}{Theorem}

\def\mb{\mathbf}

\def\mc{\mathcal}

\begin{document}
\title{ Single-Bit Consensus with Finite-Time Convergence: Theory and Applications}
\author{Mohammadreza Doostmohammadian$^\ast$, \textit{Member, IEEE}

\thanks{
Mechanical Engineering Department, Semnan University, Semnan, Iran \texttt{doost@semnan.ac.ir}.}}
\maketitle

\begin{abstract}
	In this brief paper, a new consensus protocol based on the sign of innovations is proposed. Based on this protocol each agent only requires single-bit of information about its relative state to its neighboring  agents. This is significant in real-time applications, since it requires less computation and/or communication load on agents. Using Lyapunov stability theorem the convergence is proved for networks having a spanning tree. Further, the convergence is shown to be in finite-time, which is significant as compared to most asymptotic protocols in the literature. Time-variant network topologies are also considered in this paper, and final consensus value is derived for undirected networks. Applications of the proposed consensus protocol in (i) 2D/3D rendezvous task, (ii) distributed estimation,  (iii) distributed optimization, and (iv) formation control are considered and significance of applying this protocol is discussed. Numerical simulations are provided to compare the protocol with the existing protocols in the literature.
	
	\textit{Index Terms} -- Consensus, Finite-time convergence, Lyapunov stability theorem, Graph theory, Rendezvous, Formation control
\end{abstract}

\section{Introduction} \label{sec_intro}
\IEEEPARstart{C}{onsensus} protocols  \cite{olfati_rev,bauso2006,Scientia2011,wang2010finite,liu2015finite,zuo2014new} have recently found applications in many disciplines including control \cite{nuno-suff.ness} and signal processing \cite{sayin2013single,das2016consensus,jstsp,icassp13} literature. For consensus, a group of sensors/estimators/agents reach an agreement on state values, where state may represent different quantities and parameters of interest; for example, the state may be the velocity of Unmanned Aerial Vehicles (UAVs) while moving as a flock \cite{olfati_rev} or it may represent the temperature/wind-speed in geographical fields to be estimated \cite{das2016consensus}. In this scenario, a sensor network reaches consensus on measurements and/or state innovations to estimate the underlying system state. Also, in distributed detection agents share information to reach consensus on Log-Likelihood Ratio (LLR) \cite{zhu2018distributed}  or scalar-valued decision statistic \cite{sahu2017recursive}.  
One possible application is in distributed estimation \cite{das2016consensus,jstsp,icassp13,doostmohammadian2019cyber,globalsip14} with further application in spacecraft attitude estimation \cite{crassidis2003unscented,roy1991decentralized}.  
Further, in control literature, consensus finds applications in
distributed optimization \cite{nedic2014distributed}, flight formation \cite{kia2019tutorial,olfati2002distributed,padhi2014formation,zou2012distributed}, and multi-agent rendezvous \cite{cortes2006robust,ren2007information}. Of particular interest in  aerospace system applications, along with rendezvous in 3D space, are distributed  target tracking \cite{ennasr2016distributed} and formation control via a group of UAVs. For example, in distance-based formation control the consensus protocol is used to  stabilize the UAVs to form a specific geometric shape (see Section~\ref{sec_formation} for more information).

This paper proposes a single-bit consensus with ability to converge  in finite-time. The main feature of this consensus protocol is the single-bit information update. The consensus protocol is proposed based on the sign of difference between state values. This implies that only single-bit of information (the sign) is required to update the state of the agents. This reduces the amount of  processing/computation load and/or network communication load at each agent.\textit{ It should be emphasized  that in applications with real-time data processing where the computation and communication are required  in faster time scale, the less computation and/or communication load is a significant merit.} Since the protocol is nonlinear, a new Lyapunov function is proposed to prove consensus stability. Further, it is proved that for the proposed single-bit consensus protocol the Lyapunov function vanishes in finite-time, implying the finite-time convergence of the consensus protocol.  In \cite{wang2010finite,liu2015finite,zuo2014new} finite-time consensus protocols are proposed, however these protocols  impose large amount of computation on agents and are computationally less efficient than the proposed single-bit protocol in this paper. 
 
Quantized consensus \cite{zhu2018distributed,reisizadeh2018quantized,zhu2015quantized} is a related concept, where the agents reach consensus on quantized information with finite quantization levels for possibly unbounded data. In this direction \cite{zhu2015quantized} investigates a deterministic quantization based on alternating direction method of multipliers (ADMM). In \cite{zhu2018distributed} authors adopt a single-bit quantized consensus method for detection based on Bayesian criterion and Neyman-Pearson criterion.  In \cite{sahu2017recursive} the authors propose two \textit{consensus+innovation} type distributed detectors based on Generalized Likelihood-Ratio Test (GLRT) for composite hypothesis testing via a group of sensors.
Further, communication in multi-agent systems based on single-bit of information is also adopted in distributed detection \cite{ciuonzo2017generalized,ciuonzo2017distributed}. In these works a group of sensors are spatially distributed over a surveillance field to locally detect the existence of an uncooperative target and then communicate their single-bit decisions to a fusion center. The single-bit decisions are based on either the hybrid combination of GLRT and Bayesian estimation \cite{ciuonzo2017distributed} or Generalized-Rao Test \cite{ciuonzo2017generalized}. The fusion center combines  the received information based on the fusion rules and makes a global decision.

The main contributions of this paper are as follows: (i) the proposed protocol is based on single-bit of information, which makes it practical in real-time system applications. This is the most important feature of our proposed protocol. (ii) a new Lyapunov function is proposed to prove the stability and convergence of this nonlinear consensus protocol under certain connectivity condition. This Lyapunov function is irrespective of the consensus protocol dynamics, and therefore, might be used for stability analysis of other nonlinear consensus protocols in the literature \cite{bauso2006,wang2010finite}. (iii) the convergence time of the consensus is in finite-time while reducing the computation load on agents in contrast to most asymptotic consensus protocols in the literature \cite{olfati_rev,bauso2006}. It should be noted, although finite-time consensus protocols are already exist in the literature, to name a few \cite{wang2010finite,liu2015finite,zuo2014new}, their main drawback is their computational complexity as compared to the proposed protocol in this paper.

The rest of the paper is organized as follows: Section~\ref{sec_cons} formulates  the new consensus protocol. Section~\ref{sec_proof} provides the  proof of consensus and convergence based on Lyapunov stability. Section~\ref{sec_topology} provides the convergence condition in case of time-variant switching network topologies.  Section~\ref{sec_app} provides some applications of the proposed protocol. Section~\ref{sec_sim} presents simulation  to verify the results, and finally Section~\ref{sec_conc} concludes the paper.

\section{New Consensus Protocol} \label{sec_cons}
Assume a network of  $n$ agents with ability to process information and communicate with neighboring agents to share information. The communication network of agents is represented by graph $\mc{G}=(\mc{V},\mc{E}, W)$, where $\mc{V} = \{1,...,n\}$ represents the set of graph nodes (agents), $\mc{E}$ represents the set of edges (communication links)  defined as $\mc{E}=\{(i,j)|W_{ji} \neq 0\}$. Note that $W_{ji} \neq 0$  is the weight assigned to communication link from node $i$ to node $j$. Further, the neighborhood of agent $i$ is defined as $\mc{N}_i=\{j \in \mc{V}| W_{ij} \neq 0\}$

The state of each agent is represented by $x_i$ and  $\mb{x}=[x_1,...,x_n]^\top$ represents the state of all agents. The following consensus protocol is proposed to update the state of agent $i$ as:
\begin{eqnarray} \label{eq_cons}
\dot{x}_i = \sum_{j=1}^{n} W_{ij} \text{sgn}(x_j-x_i)= \sum_{j \in \mc{N}_i} W_{ij} \text{sgn}(x_j-x_i)
\end{eqnarray}
where $\text{sgn}()$ is the sign function defined as:
\begin{eqnarray}
\text{sgn}(x)=
\begin{cases}
1       & \quad \text{if } x>0\\
-1  & \quad \text{if } x<0\\
\end{cases}
\end{eqnarray}
where the $|x_j-x_i|$ represents the absolute value of $(x_j-x_i)$. \textit{Notice that consensus protocol \eqref{eq_cons}
only requires the sign of $(x_j-x_i)$, which can be defined by single-bit of information.}

Extending the scalar-state  protocol \eqref{eq_cons} to vector-state $\mb{x}$, the updating law is as follows:
\begin{eqnarray} \label{eq_cons_vect}
\dot{\mb{x}}_i = \sum_{j \in \mc{N}_i} W_{ij} \frac{\mb{x}_j-\mb{x}_i}{\|\mb{x}_j-\mb{x}_i\|}
\end{eqnarray}
where  $\|.\|$ represents the Euclidean norm of the vector. In this case, the agents use the weighted summation of the \textit{unit vector} of its state relative to its neighbors' states for control update.

In both cases of scalar-valued consensus \eqref{eq_cons} and vector-state consensus \eqref{eq_cons_vect} the amount of information exchange and/or the computation on agents is less than the common consensus protocols in the literature \cite{olfati_rev,bauso2006,wang2010finite,liu2015finite,zuo2014new}. In protocol \eqref{eq_cons} only the sign of relative states $x_j-x_i$ is needed to be exchanged among agents and to be computed for state update.  Similarly, for protocol \eqref{eq_cons_vect} only unit vector in the direction of relative state vector is needed for computation and communication. This is the key feature reducing the amount of information exchange and/or computational load at each agent and improving the real-time feasibility of the protocol.

Now the question is how the agents exchange information on the sign function or unit vector. This depends on the nature of agents' states. Assume the state represents the position or velocity of the agent. For consensus on scalar-valued velocities, for example in flocks or vehicle platooning \cite{pirani2018cooperative}, following protocol \eqref{eq_cons} each agent only needs to know if the other agent moves faster or slower, without needing to exactly know the velocity of the neighboring agent by communication or sensing the exact velocity. For consensus on position vectors as in \eqref{eq_cons_vect} each agent uses the unit vector in the direction of relative positions of neighboring agents and, in contrast to protocols in the literature, there is no need to communicate exact positions of agents. This can be done, for example, by omni-directional cameras on agents without need to communicate exact locations (See more explanations on this in Section~\ref{sec_app_rend}). Then, the agents update their state (position) based on the weighted summation of these unit vectors. In scalar-state case, the state of agent $i$ is updated based on the weighted summation of $1$s and $-1$s, where $1$ and $-1$ are respectively assigned to the case $x_j>x_i$ and $x_j<x_i$. The agents' states get updated and evolve in time until all the agents have the same state and reach consensus. It should be mentioned that this protocol does not fail for static states. In other words, when the weighted sum of $1$s and $-1$s or the unit vectors is zero, the state of agent does not change. The state remains unchanged until the summation changes due to change in the state of neighboring agents, or the system reaches consensus and the state of all agents remains unchanged and equal.

One drawback of the given protocol \eqref{eq_cons}, and in general any non-Lipschitz protocol, is the sensitivity to time-delay. In case there is time-delay in the information exchange among agents, undesirable oscillations in agents' states may occur which is known as \textit{chattering phenomenon}. This is a side-effect of using non-Lipschitz function and is prevalent in  finite-time convergent consensus protocols as in \cite{wang2010finite,liu2015finite,zuo2014new} and also in \textit{Sliding Mode Control (SMC)} \cite{nonlin}. One solution to avoid such phenomenon is to use smooth Lipschitz functions around the equilibrium, for example \textit{saturation function},
\begin{eqnarray}
	\text{sat}(x)=
	\begin{cases}
	1   & \quad \text{if } x>a\\
	x   & \quad \text{if} -a<x<a\\
	-1  & \quad \text{if } x<-a
	\end{cases}
\end{eqnarray}	
This is proposed in SMC  as described in \cite{nonlin}. In such case the agents' states reach a convergence ball (of radius $a$) around the equilibrium in finite-time, however, the convergence inside this convergence ball is asymptotic. In terms of information exchange, the agents share single-bit of information outside this convergence ball, while inside this ball they need to share full-state information. For example, when the state represents location, in case agents' states get closer to each-other the agents are able to share more information, while in distant states only single-bit of information is exchanged. It should be mentioned, replacing the non-Lipschitz function with a Lipschitz equivalent only alleviates the effect of time-delay and does not completely vanish the chattering phenomenon.  

\section{Proof of Finite-Time Convergence} \label{sec_proof}
Here, we answer the following question: what is the connectivity requirement on the network such that state of all agents reach the same value? i.e., $\mb{x}^*=\alpha \mb{1}$ is the stable equilibrium point of the protocol \eqref{eq_cons} under what connectivity condition. To answer this, we introduce the concept of \textit{spanning tree} in directed graphs. Define a directed tree as a directed graph where every node (except the \textit{root} node) has exactly one incoming edge. The root node (also referred as \textit{leader} node) has no incoming edge. A graph has a spanning tree if it contains a directed tree as a subgraph that spans all nodes.

\begin{theorem}
	Protocol \eqref{eq_cons} reaches consensus  if and only if the communication network $\mc{G}$ has a spanning tree.
\end{theorem}
\begin{proof}
	Contradiction is used for the proof.
	\textit{Sufficiency:} if the graph $\mc{G}$ has a spanning tree, we prove that the equilibrium point of \eqref{eq_cons} is in the form $\mb{x}^*=\alpha \mb{1}$. As a contradiction assume that $\mb{x}^* \neq \alpha \mb{1}$. Therefore, consider the agent $i$ with maximum (or minimum) state. Since the network has a spanning tree there is at least one agent $j$ in the neighborhood of $i$ (or agent $i$ is in the neighborhood of agent $j$) \cite{diestel2017graph}. Therefore,$\dot{x}_i =  \sum_{j \in \mc{N}_i} W_{ij} \text{sgn}(x_j-x_i)<0$
	or $\dot{x}_i =  \sum_{i \in \mc{N}_j} W_{ji} \text{sgn}(x_i-x_j)>0$
    which both cases contradict the definition of equilibrium point.	
	\textit{Necessity:} If no spanning tree is contained in the communication graph $\mc{G}$, it implies that there is no information flow (directed path) at least among two agents. In graph theory, this implies that either the graph has at least two roots or the graph contains at least two unconnected components \cite{diestel2017graph}. In first case, note that $\dot{x}_{root}=0$ since it has no incoming information ($\mc{N}_{root}=\emptyset$ \cite{diestel2017graph}). Therefore, the states of two root agents remain the same initial values without updating, and these two agents never reach consensus. In the second case, since there is no information flow (directed path) between two components, each component reaches a consensus value which  in general  differs from the consensus value of the other component. Therefore, for both cases the consensus may not be reached.
\end{proof}	
\begin{theorem} \label{thm_stab}
	In protocol \eqref{eq_cons}, having a spanning tree in $\mc{G}$, states of all agents converge to stable consensus equilibrium point $\mb{x}^*=\alpha \mb{1}$ in finite-time.
\end{theorem}
\begin{proof}
	We prove the theorem using Lyapunov stability theorem. Define the following Lyapunov function:
	\begin{eqnarray} \label{eq_V}
	V =  x_{\max}-x_{\min}
	\end{eqnarray}	
	where $x_{\max}$ and $x_{\min}$ are respectively the maximum state value and the minimum state value of the agents, i.e. $
	x_{\max} = \max\{x_i, i \in \{1,\dots,n\}\}$ and
	$x_{\min} = \min\{x_i, i \in \{1,\dots,n\}\}$	
	In fact, $x_{\max}$ and $x_{\min}$ are time-dependent, i.e.,  the agent possessing the max/min value differs at every time instant.
	Notice that $x_{\max}=x_{\min}$ implies that the max value and min value of all agents are equal and therefore the consensus equilibrium point  $\mb{x}^*=\alpha \mb{1}$ is reached. Note that the Lyapunov function $V$ is continuous, regular, and Lipschitz. Also, $V$ is globally positive definite, i.e. $V>0 \Leftrightarrow \mb{x}^* \neq \alpha \mb{1}$ and $V=0 \Leftrightarrow \mb{x}^*=\alpha \mb{1}$. Further, Lyapunov function $V$ is radially unbounded, i.e. $V \rightarrow \infty$ as $|x_i| \rightarrow \infty$. For convergence and stability we  prove that $\dot{V}$ is negative definite.
	\begin{eqnarray}
	\dot{V} =  \dot{x}_{\max}-\dot{x}_{\min}
	\end{eqnarray}
	\begin{eqnarray}
	\dot{x}_{\max}= \sum_{j \in \mc{N}_{\max}} W_{\max j} \text{sgn}(x_j-x_{\max}) \\
	\dot{x}_{\min}= \sum_{j \in \mc{N}_{\min}} W_{\min j} \text{sgn}(x_j-x_{\min})
	\end{eqnarray}	
	Define $\min\{W\}$ as the minimum positive consensus weight of agents in weight matrix $W$. Since the weight matrix $W$ might be time-variant, the term $\min\{W\}$ might be assigned to different agents over time.  We have,
	\begin{eqnarray}
	\dot{x}_{\max}=- \sum_{j \in \mc{N}_{\max}} W_{\max j} \leq -\min\{W\} \\
	\dot{x}_{\min}= \sum_{j \in \mc{N}_{\min}} W_{\min j}\geq \min\{W\}
	\end{eqnarray}
	Therefore, $	\dot{V} \leq -2\min\{W\}$
	This implies that $\dot{V}$	is globally negative definite, i.e. $\dot{V}<0 \Leftrightarrow \mb{x}^* \neq \alpha \mb{1}$ and $\dot{V}=0 \Leftrightarrow \mb{x}^*=\alpha \mb{1}$. Therefore, based on Lyapunov stability theorem \cite{nonlin} the consensus point is globally stable equilibrium of protocol \eqref{eq_cons}.
	
	Further let $t_{cons}$ be the convergence time of the consensus protocol \eqref{eq_cons}.
	\begin{eqnarray}
	\int_{0}^{t_{cons}}\dot{V} \leq -\int_{0}^{t_{cons}}2\min\{W\} \\
	V(t_{cons})-V(0) \leq -2\min\{W\}t_{cons} \\
	{t_{cons}}\leq  \frac{x_{\max}(0)-x_{\min}(0)}{2\min\{W\}}
	\end{eqnarray}
	representing finite-time upperbound on conevergence.		
\end{proof}
One point to be noted in the proof of Theorem~\ref{thm_stab} is on the notation of $x_{max}$, $x_{min}$, $\mc{N}_{\max}$, $\mc{N}_{\min}$, and  $W_{\max j}$. For these terms max/min values do not necessarily concern a single agent over time, but these max/min values concern all agents. In other words, the agent possessing the min/max value, its neighbors, and the associated weights change over time, and therefore, the time-evolution of the Lyapunov function \eqref{eq_V} is not necessarily smooth.

It should be noted, the proof of the stability and convergence for the vector-state protocol \eqref{eq_cons_vect} follows similar Lyapunov analysis. In vector-state problem, the Lyapunov function can be considered as the perimeter of the convex hull containing the vector state of the agents, or the circumference of the smallest covering ball/circle enclosing the vector states. Following similar analysis as in above, it can be proved that the Lyapunov function is always decreasing under protocol \eqref{eq_cons_vect}.

\section{Time-Variant Network Topologies} \label{sec_topology}
Note that the consensus network of agents may change in time due to failure  or addition of new links among agents. This may particularly happen in network of mobile agents where the communication range of agents are limited or in real world applications due to obstacles. The objective of this section is to determine the conditions on changing network topology $\mc{G}$ for which the consensus can be reached. The main point in this section is that our proposed Lyapunov function does not depend on the graph topology $\mc{G}$.

\begin{theorem} \label{thm_switch}
	Consider the network topology of agents to be selected from the finite set of graphs $\Gamma_m = \{\mc{G}_1,\mc{G}_2,...,\mc{G}_m\}$, where $\mc{G}_d=(\mc{V},\mc{E}_d, W_d), d \in \{1,...,m\}$. Agents reach consensus under protocol \eqref{eq_cons} if for a sufficient sequence of bounded non-overlapping time-intervals $[t_k,t_k+l_k],~k=1,2,...$, the combination of network topologies across each time-interval contain a spanning-tree.
\end{theorem}
\begin{proof}
	Again consider the proposed positive definite Lyapunov function $V =  x_{\max}-x_{\min}$ which is independent of network topology. The proof is similar to the proof of Theorem~\ref{thm_stab}. Note that in every time-interval $[t_k,t_k+l_k]$ the combination of graph topologies contain a spanning tree. Therefore, the agent with $x_{\max}$ (or $x_{\min}$) has at least one neighbor $j$ or is a neighbor of agent $j$ in a sub-domain of the interval (not necessarily in the entire time-interval). This implies that for this time domain $\dot{x}_{\max}\leq -min\{W\}$ (or $\dot{x}_{\min}\geq min\{W\}$).  Consequently, following the statement of the proof in Theorem~\ref{thm_stab}, in this time domain $\dot{V}$ is negative definite and more precisely $\dot{V} \leq -2\min\{W\} $. This implies that after sufficient (finite) number of time-intervals the Lyapunov function reaches $0$ and consensus is achieved.
\end{proof}


\section{Applications} \label{sec_app}
\subsection{Rendezvous in 2D/3D space} \label{sec_app_rend}
In rendezvous problem \cite{cortes2006robust,ren2007information}, the goal is to devise control strategies on a group of mobile agents to eventually move them to a single location
The state of each agent is its position in 2D/3D space, and the aim is to reach a consensus on the position.
In words, each agent applies the weighted summation of the \textit{unit vector} relative to its neighbors'  positions to update its own location. In fact, using protocol \eqref{eq_cons_vect}, every agent only needs to be informed of the direction of the neighboring agent's relative position vector, but not its magnitude.
This is significant as by using, for example, omni-directional cameras \cite{kato2014localization} each agent finds information on the relative direction towards its neighbor's position, and there is no need to communicate the exact location of the agents. This approach can be implemented, for example, to improve the experimental results  in \cite{ren2008experimental} in terms of \textit{real-time} communication and computation; each robot only needs to find the direction that the neighboring robots are located using an omni-directional camera and there is no need to communicate its position to the neighboring robots. This, further, can be extended to the 3D case to implement the rendezvous task over a network of UAVs.

\subsection{Distributed estimation}
In single time-scale distributed estimation \cite{das2016consensus,jstsp,icassp13,doostmohammadian2019cyber} the idea is to track the state of the dynamical system via a network of agents. 
Consider a noisy system monitored by noise-corrupted measurements,
\begin{eqnarray}\label{eq_sys1}
\mb{x}^{k+1} = A\mb{x}^k + \mb{v}^k,\\ \label{eq_H_i}
\mb{y}^k_i = H_i\mb{x}^k + \mb{r}^k_i.
\end{eqnarray}
In the above formulation, $A$ is the dynamical system matrix, $k$ is the time-step, $y^k_i$ is the measurement of agent $i$ at time-step $k$, and ${v}^k$ and ${r}^k_i$ are the noise terms. Following the distributed estimation protocol in \cite{jstsp,jstsp14,acc13,doostmohammadian2017recovery,asilomar11}, the following  protocol based on the consensus-update law \eqref{eq_cons_vect} is considered:
%

\begin{eqnarray} \nonumber
\tilde{\mb{x}}_i^{k} &=& \sum_{j\in\mathcal{N}_\beta(i)} W_{ij}A\widehat{\mb{x}}_j^{k-1} \\  \label{eq_est}
\widehat{\mb{x}}_i^{k} &=& \tilde{\mb{x}}_i^{k} + K_i^k \sum_{j\in \mc{N}_\alpha(i)} H_j^\top\left(\frac{\mb{y}^k_j- H_j \tilde{\mb{x}}_i^{k}}{\|\mb{y}^k_j- H_j \tilde{\mb{x}}_i^{k}\|}\right).
\end{eqnarray}

where,
$\mc{N}_\alpha(i)$ and $\mc{N}_\beta(i)$ represent some specific neighborhoods of agent $i$, and $K^k_i$ is the estimation gain at agent $i$ (see the previous works by author \cite{jstsp,jstsp14,acc13,doostmohammadian2017recovery,asilomar11} for more information). One interesting extension to distributed estimation protocol \eqref{eq_est} may be considered for the case of sensor failure and countermeasures to recover for that (see more information in \cite{doostmohammadian2017recovery}.) 

One application of the above distributed estimation scenario is in target tracking based on time-difference-of-arrival (TDOA) via a group of UAVs \cite{ennasr2016distributed}. In this framework a group of UAVs  estimate the location of a mobile target based on the time-difference-of-arrival of some beacon signal received by the UAVs. Each UAV shares the TDOA-based information with its neighboring UAVs, and also shares the estimated position of the target. Then,  by consensus averaging of the position estimates and information fusion on the state-predictions, each UAV in the network can localize the source, and the group tracks the location of the mobile target.

\subsection{Distributed optimization}
In distributed optimization problem \cite{nedic2014distributed} the objective is to distributively solve the following optimization problem via a multi-agent network,
\begin{equation} \label{eq_min}
\min
\limits_{x \in \mathbb{R}} ~~ f(x) := \frac{1}{n} \sum_{i=1}^{N} f_i(x)
\end{equation}
where $f(x)$ is the continuously convex objective function and $f_i(x)$ is the local objective function only known by agent $i$. The distributed  gradient-descent-based solution for this problem applying the consensus protocol \eqref{eq_cons} is as follows,
\begin{equation} \label{eq_min2}
x^{k+1}_i = x^k_i - \alpha^k \nabla f_i(x_i^k) + \gamma \alpha^k \sum_{j \in \mc{N}_i} W_{ij} \text{sgn}(x^k_j-x^k_i)
\end{equation}
where $\gamma>0$ is a scalar weight and $\alpha^k$ is the optimization step-size at time $k$.

The protocol \eqref{eq_min2}  reduces the amount of information processing and computation load on agents as compared to \cite{nedic2014distributed}, particularly in large scale applications. Note that  each agent only relies on the sign of relative state of neighboring agents.
In the same line of research, \cite{reisizadeh2018quantized} considers the case that the communicated decision information among agents are quantized in order to alleviate the communication bottleneck in distributed optimization. The authors propose a Quantized Decentralized Gradient Descent (QDGD) and prove the convergence of their protocol for strongly convex and smooth local cost functions.

\subsection{Formation control of UAVs} \label{sec_formation}
One application within the aerospace systems is cooperative control of UAV formation \cite{kia2019tutorial,olfati2002distributed,padhi2014formation,zou2012distributed}. A group of autonomous UAVs form a pre-specified formation setup (e.g. a star-shaped formation) in which two neighboring UAVs $i$ and $j$ are in the distance $d_{ij}$ based on their formation shape. The control input to the group of UAVs are designed such that they can move along on this formation without collision in order to, for example, track a mobile target or avoid an obstacle. One example of such formation control scenario is given as following.
Define $\mb{x}_i$ as the position of UAV (or agent) $i$ in 3D space evolving in time as:
\begin{eqnarray} \nonumber
\dot{\mb{x}}_i &=& \mb{p}_i \\\label{eq_formation}
 \dot{\mb{p}}_i &=& \sum_{j \in \mc{N}_i} W_{ij}\left(\frac{\mb{x}_j-\mb{x}_i}{\|\mb{x}_j-\mb{x}_i\|}\right)-\bar{u}\lambda_2\sigma(\mb{p}_i).
\end{eqnarray}
where $\bar{u}>0$, $\sigma(\mb{y})=\mb{y}/\sqrt{(1+\|\mb{y}\|^2)}$, and
\begin{eqnarray}
W_{ij} := \frac{\bar{u}\lambda_1}{|\mc{N}_i|}\sigma(\|\mb{x}_i-\mb{x}_j\|-d_{ij})
\end{eqnarray}
with  $0<\lambda_1,\lambda_2<1$, $\lambda_1+\lambda_2=1$.
Note that the above formation protocol is a distance-based approach, where the final formation of agents depend on the rigidity of the neighboring  graph, and is also based on the distance $d_{ij}$ of every two agents $i$ and $j$ in forming the geometric shape.

\section{Simulation} \label{sec_sim}
For simulation we consider network of $10$ agents with random initial states in $[0,10]$. 
Assume the network is time-variant and switches every $0.4$ seconds between the graph topologies shown in Fig.\ref{fig_graph}.
\begin{figure} [hbpt!]
	\centering
	\includegraphics[width=3.45in]{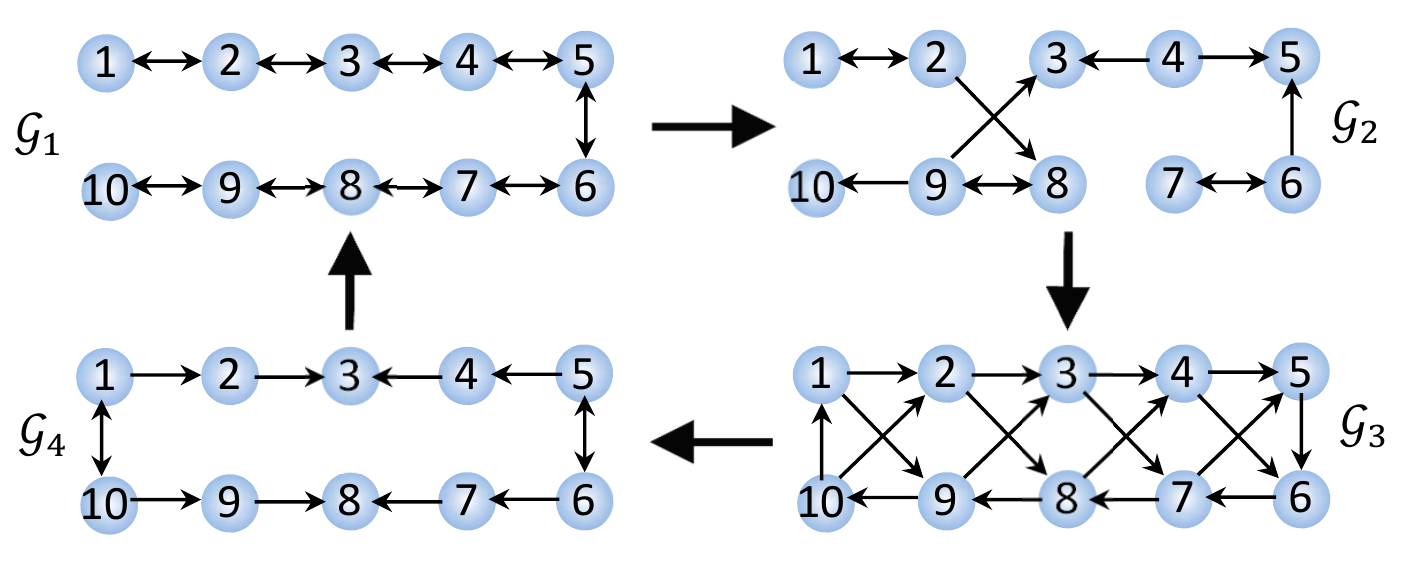}
	\caption{This figure shows the change in network topology of agents. The graphs $\mc{G}_1$ and $\mc{G}_3$ contain a spanning tree, while graphs $\mc{G}_2$ and $\mc{G}_4$ contain no spanning tree.}
	\label{fig_graph}
\end{figure}
As it can be seen from the figure, graph $\mc{G}_1$ is connected undirected and contains a spanning tree. Graph $\mc{G}_2$ is connected with no spanning tree.  $\mc{G}_3$ represents a strongly connected graph having a spanning tree. Finally, $\mc{G}_4$ contains no spanning tree as a subgraph. To check the conditions of Theorem~\ref{thm_switch} for consensus convergence, note that the combination of the network in time domain of every $0.8$ seconds contain a spanning tree. This implies that consensus can be reached, as it is shown in Fig.\ref{fig_x}. This figure shows that the difference of state values ($x_{\max} - x_{\min}$) is decreasing over time and all agents reach a consensus value.
\begin{figure}[hbpt!]
	\centering
	\includegraphics[width=3.6in]{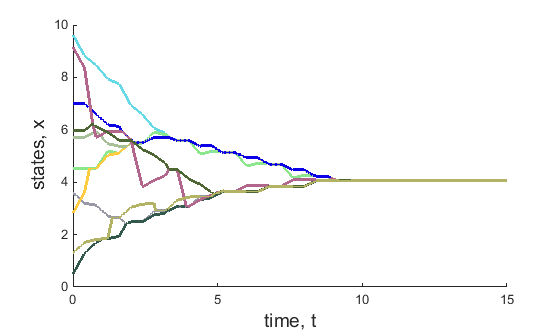}
	\caption{This figure shows the state-evolutions in time under the single-bit protocol \eqref{eq_cons} while agents' connectivity follows the switching topologies given in Fig.\ref{fig_graph}. As it can be seen $x_{max}-x_{min}$ is decreasing in time and the   consensus among agents is achieved in finite-time.}
	\label{fig_x}
\end{figure}

To compare our results, the time-evolution of the Lyapunov function \eqref{eq_V} for our proposed consensus protocol along with six other  consensus protocols from \cite{olfati_rev,bauso2006,wang2010finite,liu2015finite,zuo2014new} are shown in Fig.\ref{fig_V}. In this simulation $\alpha = 0.25$ for protocol in \cite{wang2010finite}, $\alpha = 0.5$ for  protocol in \cite{liu2015finite}, $\alpha = 0.8$, $\beta = 1.2$, $p = 3$, and $q = 5$ for protocol in \cite{zuo2014new}, $\alpha = 0.4$ for protocols in \cite{bauso2006}. All these protocols are evaluated over switching network topologies (Fig.~\ref{fig_graph}) with the same initial state values of agents. As it can be seen, the linear average consensus \cite{olfati_rev}, the geometric consensus \cite{bauso2006}, and the harmonic consensus \cite{bauso2006} all reach asymptotic stability while the convergence of the other four protocols are in finite-time.
The Lyapunov  function for our proposed protocol (and the six other protocols) is decreasing over time ($\dot{V} \leq 0$) which implies  Lyapunov stability.
\begin{figure}[hbpt!]
	\centering
	\includegraphics[width=3.6in] {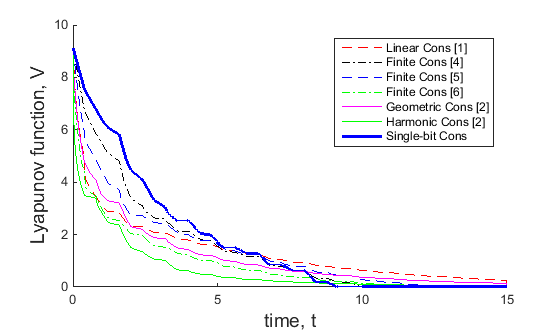}
	\caption{This figure shows the time evolution of Lyapunov function \eqref{eq_V} for seven consensus protocols: the blue solid line represents the proposed protocol (Lyapunov function of states in Fig.\ref{fig_x}) and other lines are for protocols in \cite{olfati_rev,bauso2006,wang2010finite,liu2015finite,zuo2014new}. The Lyapunov function is decreasing for all consensus protocols and reaches zero in finite-time for protocols \cite{wang2010finite,liu2015finite,zuo2014new} along with single-bit protocol and reaches zero asymptotically for protocols in \cite{olfati_rev,bauso2006}. }
	\label{fig_V}
\end{figure}

\section{Conclusions} \label{sec_conc}
Note that the protocol is based on computation and communication of single-bit of information and unit vector on state update. This makes the protocol more feasible in terms of
real-time applications, since it requires less computational load and information exchange among agents. The known protocols in the literature require the calculation of the exact relative state $x_j-x_i$ or a function  of it $f(x_j-x_i)$, however this protocol only requires the single-bit information on sign of $x_j-x_i$. In other words, the agent $i$ only needs to know for neighboring agent $j$ if $x_j<x_i$ or $x_j>x_i$. Similarly, for vector state consensus the agent $i$ only need to know the unit-vector in the direction of relative state vector, e.g. by use of omni-directional cameras in rendezvous problem as discussed in Section~\ref{sec_app_rend}. This is more real-time feasible in terms of computational complexity and communication load on agents.

The consensus protocol in this paper may be used to decentralize the information fusion in \cite{ciuonzo2017generalized,ciuonzo2017distributed}. In this case the information on detecting the existence or absence of the target may be shared by agents in their neighborhood via an undirected communication network and agents eventually average the received information via the proposed protocol and reach a consensus on detecting the target. However, this approach may result in performance degradation as compared to advanced GLRT or G-Rao fusion methods.

\bibliographystyle{IEEEbib}
\bibliography{bibliography}
\end{document}